\documentclass[sigconf]{acmart}

\usepackage{graphicx}
\usepackage{balance}
\usepackage{amsmath}
\usepackage{amssymb}
\usepackage{amsfonts}
\usepackage{algorithm}
\usepackage[noend]{algpseudocode}
\usepackage{textcomp}
\usepackage{times}
\usepackage{tikz}
\usepackage{paralist}
\usepackage{caption}
\usepackage{booktabs}
\usepackage{multirow}
\usepackage{subcaption}
\usepackage{url}
\usepackage{xspace}
\usepackage{todonotes}
\usepackage{relsize}
\usepackage{xcolor}
\usepackage[normalem]{ulem}
\usepackage{array}
\usepackage{enumitem}

\setlist{nosep,leftmargin=*}

\usetikzlibrary{shapes.misc}
\usetikzlibrary{shapes.geometric}
\usetikzlibrary{shapes,arrows}
\usetikzlibrary{positioning}

\acmConference[CIKM'20]{CIKM}{2020}{}
\newcommand{\name}{\textsc{HUKA}\xspace}

\renewcommand{\comment}[1]{}

\newtheorem{definition}{Definition}
\newtheorem{lemma}{Lemma}

\definecolor{applegreen}{rgb}{0.55, 0.71, 0.0}
\definecolor{alizarin}{rgb}{0.82, 0.1, 0.26}
\definecolor{ao(english)}{rgb}{0.0, 0.5, 0.0}

\tikzset{
itria/.style={
  draw,shape border uses incircle,
  regular polygon,
        regular polygon sides=3,shape border rotate=0,yshift=-0.5cm, scale=0.38},
}
\tikzset{fontscale/.style = {font=\relsize{#1}}}

\newcommand{\dbpedia}{DBpedia\xspace}
\newcommand{\yago}{YAGO2\xspace}
\newcommand{\gprom}{\textsf{GProM}\xspace}
\newcommand{\provsql}{\textsf{ProvSQL}\xspace}
\newcommand{\neo}{\textsf{Neo4j}\xspace}
\newcommand{\tripleprov}{\textsf{TripleProv}\xspace}

\newcommand{\potm}{\ensuremath{\mathcal{PM}}\xspace}
\newcommand{\pmo}{\ensuremath{\potm_{1:1}}\xspace}
\newcommand{\pmm}{\ensuremath{\potm_{1:m}}\xspace}
\newcommand{\mmp}{MultiMap\xspace}
\newcommand{\regq}{\ensuremath{{Q}_\text{reg}}\xspace}
\newcommand{\mmq}{\ensuremath{{Q}_\text{mm}}\xspace}

\newcommand{\lit}{\textit}
\newcommand{\labl}{\textit}

\usepackage{listings}
\usepackage{graphics}
\lstdefinestyle{RDF}{basicstyle=\ttfamily,
                        keywordstyle=\lstuppercase,
                        emphstyle=\itshape,
                        showstringspaces=false,
                        }
\lstdefinelanguage[RDF]{SPARQL}[]{SPARQL}{
  morekeywords={SELECT, WHERE, FILTER, PREFIX},
}
\definecolor{auburn}{rgb}{0.43, 0.21, 0.1}
\lstdefinestyle{customsparql}{
  belowcaptionskip=1\baselineskip,
  breaklines=true,
  xleftmargin=\parindent,
  language=SPARQL,
  showstringspaces=false,
  basicstyle=\footnotesize\ttfamily\color{black},
  keywordstyle=\bfseries\color{black},
  commentstyle=\footnotesize\ttfamily\color{black},
  identifierstyle=\color{black},
  stringstyle=\color{auburn},
}
\begin{document}

\title{How and Why is An Answer (Still) Correct? Maintaining Provenance in Dynamic Knowledge Graphs}

\author{Garima Gaur}
\affiliation{%
\institution{Computer Science and Engineering,}
\institution{Indian Institute of Technology, Kanpur}
\country{India}
}
\email{garimag@cse.iitk.ac.in}

\author{Arnab Bhattacharya}
\affiliation{%
\institution{Computer Science and Engineering,}
\institution{Indian Institute of Technology, Kanpur}
\country{India}
}
\email{arnabb@cse.iitk.ac.in}

\author{Srikanta Bedathur}
\affiliation{%
\institution{Computer Science and Engineering,}
\institution{Indian Institute of Technology, Delhi}
\country{India}
}
\email{srikanta@cse.iitd.ac.in}

\begin{abstract}
	Knowledge graphs (KGs) have increasingly become the backbone of many
	critical knowledge-centric applications. Most large-scale KGs used in
	practice are automatically constructed based on an ensemble of extraction
	techniques applied over diverse data sources. Therefore, it is important to
	establish the provenance of results for a query to determine how these were
	computed. Provenance is shown to be useful for assigning confidence scores
	to the results, for debugging the KG generation itself,  and for providing
	answer explanations.  In many such applications, certain queries are
	registered as standing queries since their answers are needed often.
	However, KGs keep continuously changing due to reasons such as changes in
	the source data, improvements to the extraction techniques,
	refinement/enrichment of information, and so on. This brings us to the
	issue of efficiently maintaining the provenance polynomials of complex
	graph pattern queries for dynamic and large KGs instead of having to
	recompute them from scratch each time the KG is updated.  Addressing these
	issues, we present \name which uses provenance polynomials for tracking the
	derivation of query results over knowledge graphs by encoding the edges
	involved in generating the answer. More importantly, \name also maintains
	these provenance polynomials in the face of updates---insertions as well as
	deletions of facts---to the underlying KG.  Experimental results over large
	real-world KGs such as YAGO and DBpedia with various benchmark SPARQL query
	workloads reveals that \name can be almost $50$ times faster than existing
	systems for provenance computation on dynamic KGs.
\end{abstract}

\maketitle

\section{Motivation}
\label{sec:motivation}

Use of large-scale knowledge graphs (KGs), which model the inter-relationships
between various entities that occur in real-life, has become common in many
knowledge-centric applications. Apart from the critical role they play in
web-search systems (e.g., Google Knowledge Graph, Microsoft Bing Satori, etc.)
they are used in applications such as e-governance~\cite{hendler2012us},
technical support~\cite{techsupp1}, drug administration~\cite{drugbank},
{scholarly search~\cite{teknobase}}, and many more. While some of the knowledge
graphs are carefully hand-crafted, most KGs are automatically constructed using
one or more information extraction pipelines over a variety of underlying data
sources.  As a result, the KG could contain facts obtained through different
mechanism and, thus, a query result could be generated through a combination of
facts with vastly varying \emph{provenance}. For instance, it is possible that
a fact derived from research studies is joined with a fact from web-based
user-generated sources~\cite{pharmacovigilence14}. Tracking fine-grained
provenance of individual facts in a knowledge graph alone (as done
in~\cite{w3prov}) is not sufficient, since we also need to track the provenance
of individual query results as well, i.e., which specific facts were involved
in generating an individual result. 

In this paper, we use \emph{provenance polynomials}~\cite{provenancesemirings}
to track the provenance of answers of a graph pattern query over knowledge
graphs stored in \emph{property-graph engines} such as Neo4j. While provenance
polynomials have been successfully used by relational
databases~\cite{provenancesemirings,Cheney2009ProvenanceID,senellart2018provenance,senellart2018provsql,arab2018gprom}
and RDF triple stores~\cite{tripleProv} in the past, they  are not available as
over a property graph store such as Neo4j. 

In addition to providing efficient derivation of provenance information for
query answers, we focus on making the query answers as well as their provenance
derivations \emph{easily updated} as the underlying KG is updated. The
``knowledge'' contained in most real-world KGs continually \emph{evolves over
time}. Such evolution could be due to addition of new facts obtained from novel
data sources, via \emph{knowledge-base completion} (KBC)
techniques~\cite{wang2017knowledge}, continuous fact refinement
approaches~\cite{nell}, etc. Similarly, facts could be culled from the KG when
evidence that they are incorrect/invalid is obtained. For example, consider a
(fragment of) medical knowledge graph extracted from recent medical news
regarding the novel SaRS-CoV-2 virus, one may observe that new symptoms such as
\textit{``SaRS-CoV-2 causes loss of smell''} and \textit{``SaRS-CoV-2 causes
skin rashes''} which need to be included in the medical KG. Addition of such
new symptoms can result in (re-)classifying many patients, whose symptoms were
considered benign earlier, as potential COVID-19 patients now. Note that it is
also possible to de-classify some symptoms, resulting in the deletion of the
corresponding fact(s) from the KG. 

As a result, queries such as ``return all COVID-19 positive cases and their
contacts in a hospital database'' that are very important and are registered as
\emph{standing queries}, may generate \emph{significantly different} answers
with potentially significant effect on the real-world decisions based on the
query answers.  Therefore, an important technical challenge that needs to be
addressed is to determine whether answers to a query is \emph{still} valid or
up-to-date in the face of changes to the KG on which it is evaluated, and if
so, demonstrate \emph{how so} by presenting its evidence. A typical evidence
one seeks is the provenance of how a query result was generated by identifying
database entries which were responsible for computing each result instance.
These are the \emph{why} and \emph{how} provenance of queries.  Provenance is
also important in demonstrating why the result of a query is no longer valid if
it is not so.

A simplistic way to check if a query result is still valid or not is to
re-execute all the registered queries after each update to the knowledge graph,
and generate an alert if a result of any query differs from the previous one.
Clearly, this is impractical due to potentially large number of registered
queries, the large scale of knowledge graphs, the need to materialize
potentially large volume of query results for comparison and the high frequency
of updates to the KG. 

One may also wonder if any incremental graph pattern matching algorithms such
as TurboFlux~\cite{turboflux}, IncIsoMat~\cite{isoIncMat},
SJ-Tree~\cite{sjTree}, Graphflow~\cite{graphflow}, etc. can be used to maintain
the basic graph pattern queries. However, none of these algorithms support
computing and maintaining the provenance (either why or how) of queries.
Further, the presence of projections in the queries we consider (see
Sec.~\ref{sec:background}) make simple pattern finding harder~\cite{wikimaze}.
Thus, incremental graph pattern matching algorithms are not applicable in the
settings we consider in this paper.

\subsection{Contributions}

Our contributions in this paper are:

\begin{itemize}

	\item We present an algorithmic framework called \textbf{\name}
		(\emph{maintaining How provenance under Updates to Knowledge grAphs})
		that achieves the above two objectives by incrementally maintaining the
		how-provenance of query results in the face of updates to the
		underlying knowledge graph (Sec.~\ref{sec:framework}).  Our solution is
		based on an adaptation of provenance
		polynomials~\cite{provenancesemirings} for KGs queried using SPARQL. It
		addresses both the why- and how-provenance of conjunctive queries under
		dynamic KGs where both insertion and deletion of facts can take place
		(updates to a fact are modeled using a deletion and a subsequent
		insertion).

	\item \name compactly maintains, using provenance polynomials, the results
		of all standing queries registered for maintenance, and the facts that
		are relevant for the query (Sec.~\ref{sec:registration}).  It also
		maintains information about the subqueries obtained by removing one
		edge at a time from the basic graph pattern of the query. In order to
		exploit the shared plans across all queries/subqueries in the workload,
		it merges individual subquery execution plans generated using AND-OR
		graphs~\cite{andortree} to form a single global execution tree. When
		the KG is updated, the updated query results is computed using a
		filter-and-refine paradigm that helps to quickly recompute the subquery
		results as well (Sec.~\ref{sec:update}).
 
	\item We use two large real-world knowledge graphs for empirical
		evaluation, namely, \yago~\cite{yago2} and \dbpedia~\cite{dbpedia}.
		Our results show that we can update the answer sets along with their
		provenance polynomials in about $1$ second, and can be faster than the
		baselines by almost $50$ times (Sec.~\ref{sec:expts}).

\end{itemize}

\section{Related Work}
\label{sec:related}

The need of query result \emph{provenance} to better understand the role of
underlying data in generating a materialized \emph{view}, called \emph{lineage}
earlier, was felt and fulfilled \cite{lineage} in data warehouse environments.
Query provenance is utilized for various tasks like schema debugging
\cite{debug}, reverse engineering queries \cite{reverse}, trust assessment
\cite{trust}, query computation in probabilistic databases \cite{trio-prob},
etc.

The first attempt to formalize the concept of provenance is by
\cite{whyandwhere}. They proposed a model to capture the \emph{why} and
\emph{where} provenance.  Later more models capturing different aspects of
query provenance, including \emph{when} \cite{whenprovenance}, \emph{how}
\cite{provenancesemirings}, and \emph{why not} \cite{whynotprovenance}, were
proposed. These variations provide different dimensions to understand the query
results.  The seminal work by Green et al. \cite{provenancesemirings}
introduced the \emph{how} provenance and modeled it using \emph{provenance
semirings}. They used the annotated relations to encode the provenance and
presented a framework to propagate these annotations along with the query
computation. This model proposed a symbolical representation of the derivation
process of query answers. Each derivation is represented by a polynomial with
each monomial representing a single derivation.  The \emph{how} provenance
constitutes of, along with the derivation process, the data items involved in
the derivation process. Since how-provenance subsumes why-provenance, we have
used the how-provenance model to capture the provenance. 

\subsubsection*{Query Provenance in Systems}

In practice, these models have been adopted quite well to compute the
provenance along with query evaluation. Trio \cite{trio} is the one of the
early systems to support lineage computation. The efforts to support query
processing along with provenance computation resulted in building different
systems.  These include Orchestra \cite{orchestra} that supports provenance
computation in a collaborative environment, \tripleprov~\cite{tripleProv} that
supports how provenance for linked data, \gprom~\cite{arab2018gprom}---a
successor of Perm \cite{perm}---that supports why, where and how provenance,
and, \provsql~\cite{senellart2018provsql} that supports the same.  These
systems, however, enable computing, storing and querying provenance for a
static dataset. Recently,~\cite{dynamicProv} proposed a theoretical model to
maintain the data provenance in a RDF versioning system.  Avgoustaki et
al.~\cite{sparqlUpdateProv} address the problem of computing provenance of data
generated by SPARQL INSERT operations. We, on the other hand, compute and
maintain the provenance of pre-computed query results of SPARQL conjunctive
queries.  The work that can be considered as our counterpart is by
\cite{gaur2017} on maintenance of provenance under fact deletions.  In this
paper, we focus on maintenance of provenance of query results when facts, i.e.,
edges are added.  We also present deletion handling for completeness.

\section{Background}
\label{sec:background}

\begin{figure*}
	\begin{minipage}{0.6\textwidth}
		\begin{subfigure}[t]{0.8\textwidth}
			\centering
			\begin{tikzpicture}[scale=0.6, transform shape, font={\fontsize{12pt}{12}\selectfont}]
				\node[minimum size= 2cm] (A1) at (-0.5,-2.8) {\shortstack{$\mathbf{\{coAuthor, out, q_{i}},$ \\ $\mathbf{Ooi,e_{15}.e_{16}\}}$}};
				\node[minimum size= 2cm] (A2) at (-6.5,-2.8) {\shortstack{$\mathbf{\{coAuthor, in, q_{j}},$ \\ $\mathbf{Ramakrishnan,e_{1}.e_{3}\}}$}};
				\tikzstyle{every node} = [rounded rectangle,draw]    		
				\node (P)[ line width = 0.5mm] at +(1,0.2) {PhD};
				\node (SS)[draw=alizarin, line width = 0.5mm] at + (-0.5,2) {Sarawagi};
				\node (C1)[draw=blue, line width = 0.5mm] at + (-7.5,2.5) {Sudarshan};
				\node (RR)[draw=alizarin, line width = 0.5mm] at +(-5,1.5) {Ramakrishnan};
				\node (UW)[draw=applegreen, line width = 0.5mm] at +(-4, -2.5) {U Wisc};
				\node (S1)[draw=blue, line width = 0.5mm] at +(-6, -2) {Gehrke};
				\node (S2)[draw=blue, line width = 0.5mm] at +(-4,4) {Godbole};
				\node (I1)[draw=applegreen, line width = 0.5mm] at +(-7,4.5) {IBM};
				\node (MS)[draw=alizarin, line width = 0.5mm] at +(5,1.5) {Stonebraker};
				\node (UCB)[draw=applegreen, line width = 0.5mm] at +(6.5,3) {MIT};
				\node (S3)[draw=blue, line width = 0.5mm] at +(1.5,4.5) {Carey};
				\node (S4)[draw=blue, line width = 0.5mm] at +(0, -2) {Ooi};
				\node (I2)[draw=applegreen, line width = 0.5mm] at +(5,-2.5) {NUS};
				\tikzstyle{every node} = [rounded rectangle]
				\node(R12) at +(-0.1,3.4) {coAuthor};
				\node(E12) at +(1.3,3.2) {$e_{12}$};
				\node(R10) at +(-1.6,3.7) {coAuthor};
				\node(E10) at +(-2.5,3.2) {$e_{10}$};
				\node(R9) at +(-4,3) {hadAdvisor};
				\node(E9) at +(-2.8,2.75) {$e_{9}$};
				\node(R1) at +(-7,0) {hadAdvisor};
				\node(E1) at +(-6,-0.5) {$e_{1}$};
				\node(R3) at +(-3.3,-1.5) {worksIn};
				\node(E3) at +(-4,-1) {$e_{3}$};
				\node(R4) at +(-7.2,1.75) {coAuthor};
				\node(E4) at +(-6,2.2) {$e_{4}$};
				\node(R5) at +(-2,0.2) {coAuthor};
				\node(R7) at +(1.7,0.8) {hasDegree};
				\node(R13) at +(3.8,3.7) {hadAdvisor};
				\node(R14) at +(3.5,-0.5) {hadAdvisor};
				\node(R16) at +(-0.275,-0.5) {hasDegree};
				\node(E16) at +(0,-1) {$e_{16}$};
				\node(R17) at +(6.8,2) {worksIn};
				\node(E17) at +(5.3,2.3) {$e_{17}$};
				\draw[->, line width = 1]  (S1) -- (RR);
				\draw[->, line width = 1]  (RR) -- (P) node[midway,above,fill=none] {hasDegree};
				\draw[->, line width = 1]  (RR) -- (UW);
				\draw[<->, line width = 1]  (RR) -- (C1);
				\draw[<->, line width = 1]  (RR) -- (S4);
				\draw[<->, line width = 1]  (SS) -- (RR) node[midway,below,pos=0.5,fill=none] {coAuthor};
				\draw[->, line width = 1]  (SS) -- (P);
				\draw[->, line width = 1]  (SS) -- (MS) node[midway,below,fill=none] {hadAdvisor};
				\draw[<->, line width = 1]  (SS) .. controls(-1.5,3).. (S2);
				\draw[->, line width = 1]  (S2) -- (I1) node[midway,below,fill=none] {worksIn};
				\draw[<->, line width = 1]  (SS) .. controls(0.5,3) .. (S3);
				\draw[->, line width = 1]  (S2) ..controls(-3,3).. (SS); 
				\draw[->, line width = 1]  (S3) -- (MS);
				\draw[->, line width = 1]  (S4) -- (MS);
				\draw[->, line width = 1]  (S4) -- (I2) node[midway,below,fill=none] {worksIn};
				\draw[->, line width = 1]  (S4) -- (P);
				\draw[->, line width = 1]  (MS) -- (UCB);
				\tikzstyle{every node} = [rounded rectangle]
				\draw[->, line width = 1]  (RR) -- (P) node[midway,below,fill=none] {$e_{2}$};
				\draw[->, line width = 1]  (RR) -- (S4) node[midway,below,fill=none] {$e_{5}$};		
				\draw[->, line width = 1]  (SS) -- (RR) node[midway,above,pos=0.5,fill=none] {$e_{6}$};
				\draw[->, line width = 1]  (SS) -- (P) node[midway,above,fill=none] {$e_{7}$};
				\draw[->, line width = 1]  (SS) -- (MS) node[midway,above,fill=none] {$e_{8}$};
				\draw[->, line width = 1]  (SS) .. controls(-1.5,3).. (S2);
				\draw[->, line width = 1]  (S2) -- (I1) node[midway,above,fill=none] {$e_{11}$};
				\draw[->, line width = 1]  (S3) -- (MS) node[midway,above,fill=none] {$e_{13}$};
				\draw[->, line width = 1]  (S4) -- (MS) node[midway,above,fill=none] {$e_{14}$};
				\draw[->, line width = 1]  (S4) -- (I2) node[midway,above,fill=none] {$e_{15}$};
				\draw[->, line width = 1]  (MS) -- (UCB);
				\draw [red,solid,fill=red!40,opacity=0.2, line width = 0.5mm] plot [smooth cycle] coordinates {(-4.5,-3) (-6, 1.5) (-0.5,2.5) (4.75,2.1) (6,3.5)  (7.5,3.5) (5,1) (-0.5,1.5) (-2.5,1.5) (1.25,0.6) (1.25,-0.32) (-4.25,0.6) (-3,-3) };
				\draw [blue,solid,fill=blue!40,opacity=0.2, line width = 0.5mm] plot [smooth cycle] coordinates {(-0.2,-1.1) (4.75,2.1) (6,3.5)  (7.5,3.5) (5.2,1)  (-0.5,-2.5) (-4.5,0.6) (-3,-3) (-4.5,-3) (-6, 1.5) (-4.5,2) (1.25,0.7) (1.25, -0.3) (-2.75,0.6)};
			\end{tikzpicture}    
			\caption{Knowledge Graph}
			\label{fig:kg}
		\end{subfigure}
	\end{minipage}
	\begin{minipage}{0.3\textwidth}
		\begin{subfigure}[t]{1\textwidth}
		\centering
			\begin{tikzpicture}[scale=0.7, transform shape, line width = 1, font={\fontsize{12pt}{12}\selectfont}]
				\tikzstyle{every node} = [circle, draw=black]
				\node(a) at (0,3) {?org1};
				\node (b) at (2,1.5) {?collab};
				\node (c) at (4,3.5) {?stud};
				\node (d) at (5.5,0.5) {?org2};
				\node (e) at (0, 0.5) {PhD};
				\node (f) at (6,2.5) {?prof};
				\tikzstyle{every node} = [circle]
				\draw[->]  (b) -- (a);
				\node (r1)[font=\large] at (1.5,2.5) {$worksIn$};     
				\draw[->]  (b) -- (c);
				\node (r2) [font=\large] at (3.8,2.5){$coAuthor$};     
				\draw[->]  (f) -- (d);
				\node (r3) [font=\large] at (5,1.5) {$worksIn$};     
				\draw[->]  (b) -- (e); 
				\node (r4) [font=\large] at (1.5,0.5) {$hasDegree$};  
				\draw[->]  (c) -- (f);
				\node (r5) [font = \large] at (5.7,3.4) {$hadAdvisor$};     
			\end{tikzpicture}   
			\captionof{figure}[]{Query pattern}
			\label{fig:query}
		\end{subfigure}
	\end{minipage}
	\caption{Running example of a toy Knowledge Graph of academic collaborations and a corresponding graph pattern query.}
	\label{fig:runningEg}
\end{figure*}
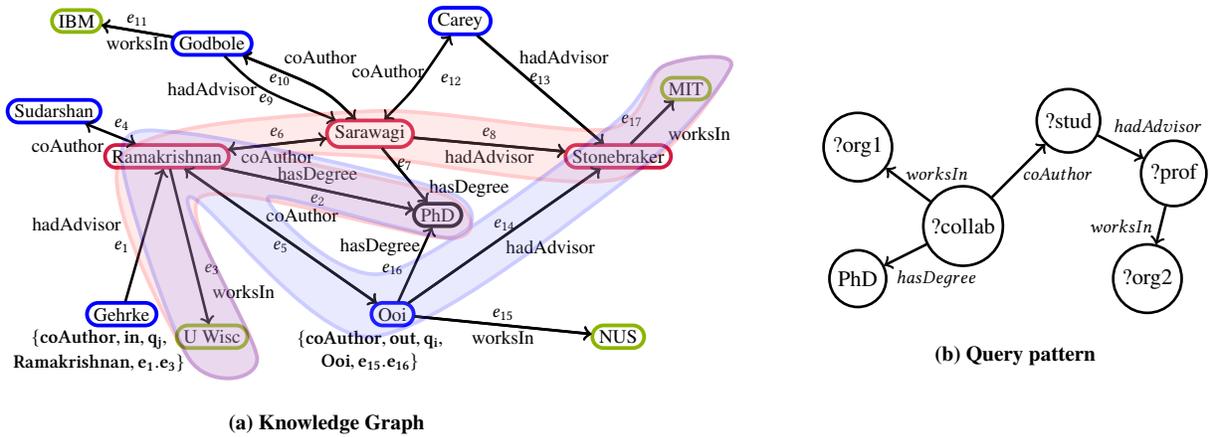

\subsubsection*{Knowledge Graph}

A knowledge graph (KG) $G \left(V, E, R, \mathcal{I}\right)$ is a directed
graph with a set of vertices $V$, a set of edges $E$, a set of relationships
$R$ and a function $\mathcal{I}: E \to \mathcal{N}$ associating each edge to a
unique identifier given as $\mathcal{I}(e) = e_i$ with $e_i \in \mathcal{N}$.
Each edge, $e\left(u,l(e),v\right)$, is labeled with $l(e)$ $(\in R)$ that
represents the relation between vertex $u$ and $v$. The complete description of
an edge is $(e_i, u, l(e), v)$. However, we will often use either $(u, l(e),
v)$ or $e_i$ to refer to a particular edge with identifier $i$. Note that this
representation closely corresponds to the RDF model  used commonly to model a
KG. Each entity or property value represented as a vertex in KG has an
associated URI or a literal, and we assume that they are assigned a unique
identifier (i.e., through $\mathcal{I}(\cdot)$). A labeled edge corresponds to
a \emph{fact}, also represented as a \emph{triple} $\langle u, l(e),v \rangle$
where $u$, $l(e)$, and $v$ are respectively the \emph{subject},
\emph{predicate} and \emph{object} of the fact.

\subsubsection*{Graph Query}

Queries on the KG are posed as graph patterns, and the task is to find
\emph{subgraphs} that match the pattern.  The graph query language SPARQL is
used to formally express a query as a collection of \emph{triple patterns}.
Similar to a triple, a \emph{triple pattern} also consists of a subject, a
predicate and an object.  The subject and object can be variables or can be
bound to a particular vertex of the KG. A predicate can also be a variable or
can be one of the labels of the edges of the KG. A graph query is formulated as
a conjunction of triple patterns and the aim is to find all possible bindings
to the variables of the triple patterns.  Thus, each edge in the query graph
corresponds to a triple pattern in the KG that needs to be matched.

In this paper, we restrict ourselves to core \emph{basic graph pattern} (BGP)
queries in SPARQL. For the fragment of SPARQL graph pattern queries we
consider, it is also known that they can be expressed using SPJ
(Select-Project-Join) fragment of relational
algebra~\cite{cyganiak2005,tripleProv}. We leave it for future work to adopt
the ideas of \cite{tripleProv,tripleProvExtension} to extend to full SPARQL. We
also do not consider path expression queries since their provenance derivation
is tantamount to enumerating all paths between vertices in a graph which is
known to be intractable~\cite{arenas2012}. 

\subsubsection*{Graph Query Answer}

The answer of a query $A(Q)$ is a collection of subgraphs of KG that
\emph{match} the query pattern $Q$. To formally define a \emph{matching}
subgraph, consider a query $Q=\{t_{1}, \dots,t_{n}\}$ as a collection of $n$
triple patterns. We define a surjective mapping $v_{S}: Q \to E_{S}$ which maps
each triple pattern of the query $Q$ to an edge of subgraph $S(V_{S}, E_{S})$.
A subgraph $S$ is said to match a query pattern $Q$, if $\forall t_{i} \in Q,
\exists e_{j} \in S \mathrm{~such~that~}v_{s}(t_{i}) = e_{j}$.  This mapping
enforces the query pattern constraints. In other words, mapping ensures that if
two triple patterns $t_{i}$ and $t_{j}$ of a query have common node(s), then
the corresponding edges $v_{S}(t_{i})$ and $v_{S}(t_{j})$ should preserve the
same restriction.

\section{The \name Framework}
\label{sec:framework}

Our goal is to maintain the results and their provenance for a set of standing
queries that have been registered for maintenance in the system, as the
underlying KG undergoes \emph{updates}, i.e., insertion/deletion of
edges/nodes.  In this section, we delineate the provenance model used and the
key idea of maintaining \emph{potential matches}.  In
Sec.~\ref{sec:registration}, we discuss the \emph{offline} phase involving
registration of a standing query.  When an update occurs, \name operates in the
\emph{online} phase as discussed in Sec.\ref{sec:update}.

Throughout the paper, we will use a running example of a toy knowledge graph
depicting relationships between academic researchers and a query pattern
illustrated in Fig.~\ref{fig:runningEg}.  Facts such as \lit{Sarawagi} had
\lit{Stonebraker} as her advisor are represented using a directed edge labeled
\lit{hadAdvisor} connecting the corresponding vertices and is stored as a
triple $\langle$\labl{Sarawagi}, \labl{hadAdvisor}, \labl{Stonebraker}$\rangle$
in RDF. There are relationships such as \lit{coAuthor}  represented using a
bi-directional edge, which correspond to two triples, e.g.,
$\langle$\labl{Sarawagi}, \labl{coAuthor}, \labl{Ramakrishnan}$\rangle$ {and}
$\langle$\labl{Ramakrishnan}, \labl{coAuthor}, \labl{Sarawagi}$\rangle$. 

We define an example graph pattern query on the above KG which seeks to return
professors as well as all collaborators of their students. Note that it only
accepts professors and collaborators currently working in an organization. It
can be written in SPARQL as follows:
\begin{lstlisting}
SELECT ?prof ?collab WHERE {
	?stud hadAdvisor ?prof.
	?prof worksIn ?org2.
	?collab coAuthor ?stud.
	?collab hasDegree PhD.
	?collab worksIn ?org1. }
\end{lstlisting}
The query graph corresponding to the above query pattern is given in
Fig.~\ref{fig:query}.

\subsection{Why and How Provenance}

\begin{table}
	\caption{Query result and provenance of the running example}
	\label{tab:queryresults}	
	\resizebox{1\columnwidth}{!}{%
		\begin{tabular}{|c|c||m{2.6cm}|m{2.2cm}|}
			\toprule
			\bf ?prof & \bf ?collab & \bf \emph{why} provenance & \bf \emph{how} provenance \\
			\midrule
			Stonebraker & Ramakrishnan & $\{\{e_{2},e_{3},e_{6},e_{8},e_{17}\},$ $\{e_{2},e_{3},e_{5},e_{14},e_{17}\}\}$ & $e_{2}.e_{3}.e_{6}.e_{8}.e_{17} + e_{2}.e_{3}.e_{5}.e_{14}.e_{17}$ \\
			\bottomrule 
		\end{tabular}
	}
\end{table}

The two important classes of query provenance, namely, the \emph{why}
provenance and the \emph{how} provenance~\cite{Cheney2009ProvenanceID} in
graphs are next illustrated using our running example KG and the graph pattern
query from Fig.~\ref{fig:runningEg}. Consider Table~\ref{tab:queryresults} that
lists the results and their provenance of running the query on the example KG. 
	
The why-provenance is simply a set of sets comprising of edges contributed in a
particular derivation of a specific result item.  The how-provenance, on the
other hand, symbolically encodes the interaction of edges involved in
derivations of an answer. As shown in Table~\ref{tab:queryresults}, the
how-provenance is expressed as a polynomial with each monomial representing
bindings of edges that lead to the answer $\langle$\labl{Stonebraker,
Ramakrishnan}$\rangle$.  Here, two different sets of edge combinations
$\left\{e_2, e_3, e_6, e_8, e_{17}\right\}$ and $\left\{e_2, e_3, e_5, e_{14},
e_{17}\right\}$ lead to the same result item. In Fig.~\ref{fig:kg}, subgraphs
corresponding to these two different edge bindings are shaded in red and blue
respectively.  Since how-provenance subsumes and captures more information
about generated results than why-provenance, \name uses the how-provenance
model represented through provenance polynomials. 

\begin{figure}
	\begin{subfigure}[t]{0.47\columnwidth}
		\centering
			\begin{tikzpicture}[scale=0.7, transform shape, font={\fontsize{12pt}{12}\selectfont}, line width=0.4mm]
				\tikzstyle{every node} = [rounded rectangle, draw,, line width = 0.5mm]
				\node[draw=blue] (x3) at +(1.5,3.5) {Carey};
				\node[draw=applegreen] (x2) at +(1,0.5) {MIT};
				\node[draw=alizarin] (x1) at +(0,2) {Stonebraker};
				\node[draw=alizarin] (x4) at +(3,2) {Sarawagi};
				\node (x5) at +(3.5,0.5) {PhD};
				\tikzstyle{every node} = [rounded rectangle]
				\draw[->] (x3) -- (x1);
				 \node(R1) at (0.5,2.75) {$hadAdvisor$};   
				\draw[<->] (x3) -- (x4);
				 \node (R2) at (3.6,2.75) {$coAuthor$};   
				\draw[->] (x1) -- (x2);
				\node (R3) at (1.2,1.3) {$worksIn$};   
				\draw[->] (x4) -- (x5);
				\node (R4) at (3.3,1.3) {$hasDegree$};   
			\end{tikzpicture}    
			\caption{$G_{1}$: a 1:1 potential match}
			\label{fig:g1}
	\end{subfigure}
		\begin{subfigure}[t]{0.47\columnwidth}
			\centering
			\begin{tikzpicture}[scale=0.7, transform shape, font={\fontsize{12pt}{12}\selectfont},line width=0.4mm]
				\tikzstyle{every node} = [rounded rectangle, draw,, line width = 0.5mm]
				\node (x2) at +(1,0.5) {PhD};
				\node[draw=alizarin] (x1) at +(0,2) {Sarawagi};
				\node[draw=alizarin] (x4) at +(4,2) {Godbole};
				\tikzstyle{every node} = [rounded rectangle]
				\draw[->] (x4) ..controls(2,2.5).. (x1);
				 \node(R1) at (2,2.7) {$hadAdvisor$};   
				\draw[<->] (x1) ..controls(2,1.5).. (x4);
				 \node (R2) at (2,1.9) {$coAuthor$};   
				\draw[->] (x1) -- (x2);
				\node (R4) at (1.5,1.1) {$hasDegree$};   
			\end{tikzpicture}    
			\caption{$G_{2}$: a 1:m potential match}
			\label{fig:g2}
		\end{subfigure}
	\caption{Result of query in Fig.~\ref{fig:kg} and its potential matches.}
	\label{fig:valid}
\end{figure}

\subsection{Potential Matches}

Fundamental to our approach is the concept of \emph{potential matches} of a
query $Q$. Intuitively, these potential matches of a query correspond to
subgraphs of the KG that partially match the given query graph pattern. We
maintain the potential matches of a query $Q$ that require \emph{only a single
edge insertion} to transform them to full matches, i.e., actual results of $Q$.
We now formally define potential matches, and prove how they can be maintained. 

\begin{definition}[Potential Match]
	Any subgraph $S$ of the knowledge graph $G$ which can become an actual match of a query $Q$
	after a \emph{single} edge insertion is called a \textbf{potential match}.\hfill $\Box$
\end{definition}

Consider a query $Q$ consisting of $n$ triple patterns.  Suppose $S \in G$ is a
subgraph such that $n-1$ triple patterns of the query $Q$ match and only $1$
triple pattern does not match.  After a new edge $e$ that matches with the
unmatched triple pattern of $Q$ is added to the $S$, the new subgraph, $S^{*} =
S \cup \{e\}$ now becomes an actual match for $Q$. In other words, $S^{*} \in
A(Q),$ where $A(Q)$ is the answer-set of the query $Q$.  $S$ is, therefore, a
potential match.

Next, we categorize the potential matches based on the number of matches
resulting from the newly inserted edge. Either the newly added edge matches a
single triple pattern of the $Q$ or it matches many triple patterns in the
query. 

\begin{definition}[1:1 Potential Match]
	If a newly added edge $e$ to a potential match $S$ matches only one triple
	pattern of the query $Q$ such that $S^{*} = S \cup \{e\}$ becomes an actual
	match of $Q$, then the subgraph $S$ is called a \textbf{1:1 potential
	match}, denoted by \pmo.\hfill $\Box$
\end{definition}	
	
\begin{definition}[1:m Potential Match]
	If a newly added edge $e$ to a potential match $S$ matches multiple triple
	patterns of the query $Q$ such that $S^{*} = S \cup \{e\}$ becomes an
	actual match of $Q$, then the subgraph $S$ is called a \textbf{1:m
	potential match}, denoted by \pmm.\hfill $\Box$
\end{definition}

We illustrate these important concepts using Fig.~\ref{fig:query}. The
subgraphs $G_{1}$ and $G_{2}$, shown in Fig.~\ref{fig:g1} and Fig.~\ref{fig:g2}
respectively, are subgraphs of the example toy KG, and are potential matches of
the parent query. Although both are potential matches, $G_1$ is anticipating
the insertion of \emph{only one} fact $qp_1^1: \langle ?collab, worksIn, ?org1
\rangle$ to match the query. On the other hand, $G_2$ needs to get match for
\emph{two} triple patterns $qp_2^1: \langle ?prof, worksIn , ?org2 \rangle$ and
$qp_2^2: \langle ?collab, worksIn, ?org1 \rangle$ in the query. 

Interestingly, a single fact insertion can complete both kinds of potential
matches. For instance, if a fact $(Sarawagi, worksIn, IITB)$ is added to the
KG, it will match to the query pattern $qp_1^1$, resulting $G_1$ to fully
satisfy the query. At the same time, the same fact also matches both $qp_2^1$
and $qp_2^2$ making even $G_2$ to satisfy the query. 

\subsection{Query Categorization}

We categorize registered queries into two types, \emph{regular} and
\emph{MultiMap}, based on the kind of  potential matches it may have. Both of
these are \emph{conjunctive} queries and they differ only in the distribution
of predicates among triple patterns. We formally define them next.

\begin{definition}[Regular Query] 
	A query that has only 1:1 potential matches (and no 1:m potential match) is
	called a \textbf{regular query}, denoted by \regq.
\end{definition}

For example, the query in Fig.~\ref{fig:query} is not a regular query as its
two triple patterns $\langle?prof, worksIn, ?org2 \rangle$ and $\langle?collab,
worksIn, ?org1 \rangle$ have the same predicate $worksIn$.

\begin{definition}[MultiMap Query]
	A query that has 1:m potential matches (possibly in addition to 1:1
	potential matches) is called a \textbf{MultiMap query}, denoted by \mmq.
\end{definition}

The example query, shown in Fig.~\ref{fig:query}, is a \mmp query with $G_{2}$,
in Fig.\ref{fig:g2}, as a \pmm.  A necessary condition for a conjunctive query
to be a \mmp query is to have at least two triple patterns with common
predicate. However, this condition is not sufficient.  We discuss those
characteristics in Sec.~\ref{subsec:queryChecker} while discussing the
registration process of \mmp queries.

The separation of potential matches and the queries into different classes is
important since they need to be handled differently. Our strategy of query
registration process is based on the following lemmas.

\begin{lemma}
	\label{obs:ob1}
	A \pmo can satisfy one and only one subquery of size $n-1$ of a given
	parent query $Q$ of size $n$.
\end{lemma}

\begin{proof}
	Consider a parent query $Q = \{ t_{1}, \dots, t_{n}\}$ of size $n$, i.e.,
	having $n$ triple patterns. A subquery $Q_{i}$ of size $n-1$ is generated
	by removing the triple pattern $t_{i}$ of $Q$, i.e., $Q_{i} = \{ t_{1},
	\dots, t_{i-1}, t_{i+1}, \dots, t_{n} \}$.

	Assume that a subgraph $S$ is a \pmo such that completing the triple $t_i$
	will complete its match to the parent query $Q$.  By definition, $S$
	satisfies the subquery $Q_i$.

	Assume that $S$ also satisfies another subquery $Q_j (j \neq i)$ of $Q$.
	This implies that $S$ has all the matching edges of $Q_j$.  Since $Q_j$
	misses only the triple pattern $t_j$, it must contain $t_i$.  This is a
	contradiction and, therefore, $S$ cannot match any other subquery $Q_j$.

	Hence, $S$ matches one and only one subquery of $Q$.
	\hfill{}
\end{proof}

\begin{lemma}
	\label{obs:ob2}
	A \pmm cannot satisfy any subquery of size $n-1$ of a given parent query
	$Q$ of size $n$.
\end{lemma}

\begin{proof}
	Consider a parent query $Q = \{ t_{1}, \dots, t_{n}\}$ of size $n$, i.e.,
	having $n$ triple patterns. A subquery $Q_{i}$ of size $n-1$ is generated
	by removing the triple pattern $t_{i}$ of $Q$, i.e., $Q_{i} = \{ t_{1},
	\dots, t_{i-1}, t_{i+1}, \dots, t_{n} \}$.

	Assume that a subgraph $S$ is a \pmm such that completing the set of
	triples $\{ t_{i_1}, \dots, t_{i_k} \}$ will complete its match to the
	parent query $Q$.

	Assume that $S$ satisfies a subquery $Q_j$ of $Q$.  Since the subquery
	$Q_j$ misses only one triple pattern, $t_j$, this implies that $S$ matches
	the rest of the $n-1$ triple patterns.  By definition, this is a
	contradiction.

	Hence, $S$ cannot match any subquery of $Q$.
	\hfill{}
\end{proof}

\begin{lemma}
	\label{obs:ob3}
	After an edge insertion, a \pmm satisfies a parent query $Q$ if and only if
	the \pmm satisfies all the subqueries of $Q$.
\end{lemma}

\begin{proof}
	(Only if) Suppose a \pmm $S$ satisfies a query $Q$ consisting of $n$ triple
	patterns. This implies that $S$ has a match for every triple pattern $t_i
	\in Q, \ i = 1, \dots, n$.  Consider any subquery $Q_i$ of $Q$.  Since it
	consists of the rest $n-1$ triple patterns except $t_i$, $S$ necessarily
	satisfies $Q_i$. 

	(If) Since $S$ satisfies all subqueries, it satisfies, in particular, two
	subqueries $Q_i$ and $Q_j$ ($j \neq i$) of $Q$.  The triple patterns $t_l,
	\ l = 1, \dots, n, \ l \neq i$ match $S$ due to $Q_i$.  Since the triple
	pattern $t_i$ match due to $Q_j$, together all the $n$ triple patterns
	match.  Hence, $S$ matches the entire query $Q$.
	\hfill{}
\end{proof}

Next, we show how these observations can be used to preprocess queries during
their registration for maintenance.

\section{Registering Standing Queries}
\label{sec:registration}

In this section, we discuss in detail the process of registering a query.
These queries are called \emph{standing queries} or \emph{continuous
queries}~\cite{standing} since their results are always relevant and need to be
updated when the KG undergoes changes.  When a standing query is registered for
maintenance, apart from evaluating its results, \name performs the following
additional steps: (a)~generation of all subqueries obtained by removing one
triple pattern at a time from the query and identification of all the potential
matches, (b)~annotation of the connection points in the knowledge graph, and
(c)~precomputation of a query processing plan for the subqueries so that they
can be recomputed efficiently. In the rest of the section, we describe these
steps in detail, first for regular queries and then for \mmp queries since they
need special handling.

\subsection{Subquery Generation}
\label{subsec:sqGenerator}

Since we are interested in only those subqueries that are just one triple
pattern less than the registered query being processed, they are all of size
$n-1$. From the running example query in Fig.~\ref{fig:query}, we can generate
a number of subqueries by systematically removing each triple pattern specified
in the query. This results in $5$ subqueries as listed in
Fig.~\ref{fig:subquery}.  Notice that in some cases the removal of a triple
pattern can lead to generation of two disconnected subqueries.

If a given query is of size $n$, i.e., has $n$ triple patterns, removal of a
triple pattern can generate at most two subgraphs in the resulting subquery.
Denoting these two subgraphs as $SQ_1$ and $SQ_2$, without any loss of
generality, we assume that $|SQ_1| = k$ and $|SQ_2| = n-k-1$ where $0 \le k \le
n-1$. The subqueries can be categorized into 4 types depending on the size of
their subgraphs and the resulting connection points. These types are
illustrated in Fig.~\ref{fig:category}. In each of these illustrations, the
gray-shaded node is part of the subgraph it is shown connected to; they have
been shown separately only to highlight the subquery generation process. 

\begin{itemize}

	\item \textbf{Type I:} These subqueries have a singleton subgraph $k = 0$,
		and are generated when the triple pattern corresponding to the edge of
		a leaf node of the query graph is removed. For instance, the subquery
		shown in Fig.~\ref{fig:worksIn1}, generated by removing \labl{worksIn}
		edge. This leads to $|A(SQ_1)|$ connection points in the knowledge
		graph. The general subquery structure for this type is given in
		Fig.~\ref{fig:type1}.

	\item \textbf{Type II:} In this case, one of the subgraphs contains only
		one triple pattern, i.e., $k = 1$, as shown in Fig.~\ref{fig:type2}. In
		this case, we may end up identifying all edges with the predicate of
		the triple pattern in $SQ_2$ as connection points which could be a
		potentially large number of instances in the knowledge graph. We avoid
		this, and instead only identify the connection points derived from
		$A(SQ_1)$, leading to only $|A(SQ_1)|$ connection points.
		Fig.~\ref{fig:hadAdvisor} shows an example.
	
	\item \textbf{Type III:} As shown in Fig.~\ref{fig:type3}, both $SQ_1$ and
		$SQ_2$ contain at least two triple patterns, i.e., $2 \le k < n-2$. For
		instance, removal of $coAuthor$ generated two subqueries, shown in
		Fig.~\ref{fig:coAuthor}, each of size $2$. This results in at most
		$|A(SQ_1)| + |A(SQ_2)|$ connection points. 

	\item \textbf{Type IV:} In this last type, the removal of a triple pattern
		from the query does not disconnect the query graph, i.e., $k = n - 1$.
		As shown in Fig.~\ref{fig:type4}, each match of $SQ_1$ can result in
		two potential matches leading to a total of $2 \cdot |A(SQ_1)|$
		connection points in the graph since they may be annotated from both
		sides.

\end{itemize}

It is important to note that in each case, both $SQ_1$ and $SQ_2$ subgraphs are
effectively treated as distinct query patterns and are maintained by \name.

Potential matches are identified as subgraphs of the KG that match the
subqueries of a parent query.

\begin{figure*}
	\begin{minipage}{0.19\textwidth}
		\begin{subfigure}[t]{1\textwidth}
			\centering
			\begin{tikzpicture}[scale=0.45, transform shape]
				\tikzstyle{every node} = [circle, draw=black, font=\Large]
				\node(a) at (0,3) {?org1};
				\node[fill=gray!30, thick] (b) at (2,1.5) {?collab};
				\node (c) at (4,3.5) {?stud};
				\node (d) at (5.5,0.5) {?org2};
				\node (e) at (0, 0.5) {PhD};
				\node (f) at (6,2.5) {?prof};
				\tikzstyle{every node} = [circle]
				\draw[->,dashed, line width=0.5mm]  (b) -- (a);
				\node (r1)[fontscale=1.5] at (1.5,2.5) { worksIn};     
				\draw[->]  (b) -- (c);
				\node (r2) [fontscale=1.5] at (3.8,2.5){coAuthor};     
				\draw[->]  (f) -- (d);
				\node (r3) [fontscale=1.5] at (5,1.5) {worksIn};     
				\draw[->]  (b) -- (e); 
				\node (r4) [fontscale=1.5] at (1.5,0.5) {hasDegree};  
				\draw[->]  (c) -- (f);
				\node (r5) [fontscale=1.5] at (5.7,3.4) {hadAdvisor};     
			\end{tikzpicture}   
			\caption{Removing worksIn}
			\label{fig:worksIn1}
		\end{subfigure}
	\end{minipage}
	\begin{minipage}{0.19\textwidth}
		\begin{subfigure}[t]{1\textwidth}
			\centering
			\begin{tikzpicture}[scale=0.45, transform shape]
				\tikzstyle{every node} = [circle, draw=black, font=\Large]
				\node(a) at (0,3) {?org1};
				\node[fill=gray!30, thick] (b) at (2,1.5) {?collab};
				\node[fill=gray!30, thick] (c) at (4,3.5) {?stud};
				\node (d) at (5.5,0.5) {?org2};
				\node (e) at (0, 0.5) {PhD};
				\node (f) at (6,2.5) {?prof};
				\tikzstyle{every node} = [circle]
				\draw[->]  (b) -- (a);
				\node (r1)[fontscale=1.5] at (1.5,2.5) { worksIn};     
				\draw[->,dashed, line width=0.5mm]  (b) -- (c);
				\node (r2) [fontscale=1.5] at (3.8,2.5){coAuthor};     
				\draw[->]  (f) -- (d);
				\node (r3) [fontscale=1.5] at (5,1.5) {worksIn};     
				\draw[->]  (b) -- (e); 
				\node (r4) [fontscale=1.5] at (1.5,0.5) {hasDegree};  
				\draw[->]  (c) -- (f);
				\node (r5) [fontscale=1.5] at (5.7,3.4) {hadAdvisor};     
			\end{tikzpicture}   
			\caption{Removing coAuthor}
			\label{fig:coAuthor}
		\end{subfigure}
	\end{minipage}
	\begin{minipage}{0.19\textwidth}
		\begin{subfigure}[t]{1\textwidth}
			\centering
			\begin{tikzpicture}[scale=0.45, transform shape]
				\tikzstyle{every node} = [circle, draw=black, font=\Large]
				\node(a) at (0,3) {?org1};
				\node (b) at (2,1.5) {?collab};
				\node[fill=gray!30, thick] (c) at (4,3.5) {?stud};
				\node (d) at (5.5,0.5) {?org2};
				\node (e) at (0, 0.5) {PhD};
				\node[fill=gray!30, thick] (f) at (6,2.5) {?prof};
				\tikzstyle{every node} = [circle]
				\draw[->]  (b) -- (a);
				\node (r1)[fontscale=1.5] at (1.5,2.5) { worksIn};     
				\draw[->]  (b) -- (c);
				\node (r2) [fontscale=1.5] at (3.8,2.5){coAuthor};     
				\draw[->]  (f) -- (d);
				\node (r3) [fontscale=1.5] at (5,1.5) {worksIn};     
				\draw[->]  (b) -- (e); 
				\node (r4) [fontscale=1.5] at (1.5,0.5) {hasDegree};  
				\draw[->, dashed, line width=0.5mm]  (c) -- (f);
				\node (r5) [fontscale=1.5] at (5.7,3.4) {hadAdvisor};     
			\end{tikzpicture}   
			\caption{Removing hadAdvisor}
			\label{fig:hadAdvisor}
		\end{subfigure}
	\end{minipage}
	\begin{minipage}{0.19\textwidth}
		\begin{subfigure}[t]{1\textwidth}
			\centering
			\begin{tikzpicture}[scale=0.45, transform shape]
				\tikzstyle{every node} = [circle, draw=black, font=\Large]
				\node(a) at (0,3) {?org1};
				\node (b) at (2,1.5) {?collab};
				\node (c) at (4,3.5) {?stud};
				\node (d) at (5.5,0.5) {?org2};
				\node (e) at (0, 0.5) {PhD};
				\node[fill=gray!30, thick] (f) at (6,2.5) {?prof};
				\tikzstyle{every node} = [circle]
				\draw[->]  (b) -- (a);
				\node (r1)[fontscale=1.5] at (1.5,2.5) { worksIn};     
				\draw[->]  (b) -- (c);
				\node (r2) [fontscale=1.5] at (3.8,2.5){coAuthor};     
				\draw[->, dashed, line width=0.5mm]  (f) -- (d);
				\node (r3) [fontscale=1.5] at (5,1.5) {worksIn};     
				\draw[->]  (b) -- (e); 
				\node (r4) [fontscale=1.5] at (1.5,0.5) {hasDegree};  
				\draw[->]  (c) -- (f);
				\node (r5) [fontscale=1.5] at (5.7,3.4) {hadAdvisor};     
			\end{tikzpicture}   
			\caption{Removing other worksIn}
			\label{fig:worksIn2}
		\end{subfigure}
	\end{minipage}	
	\begin{minipage}{0.19\textwidth}
		\begin{subfigure}[t]{1\textwidth}
			\centering
			\begin{tikzpicture}[scale=0.45, transform shape]
				\tikzstyle{every node} = [circle, draw=black, font=\Large]
				\node(a) at (0,3) {?org1};
				\node[fill=gray!30, thick] (b) at (2,1.5) {?collab};
				\node (c) at (4,3.5) {?stud};
				\node (d) at (5.5,0.5) {?org2};
				\node (e) at (0, 0.5) {PhD};
				\node (f) at (6,2.5) {?prof};
				\tikzstyle{every node} = [circle]
				\draw[->]  (b) -- (a);
				\node (r1)[fontscale=1.5] at (1.5,2.5) { worksIn};     
				\draw[->]  (b) -- (c);
				\node (r2) [fontscale=1.5] at (3.8,2.5){coAuthor};     
				\draw[->]  (f) -- (d);
				\node (r3) [fontscale=1.5] at (5,1.5) {worksIn};     
				\draw[->,dashed, line width=0.5mm]  (b) -- (e); 
				\node (r4) [fontscale=1.5] at (1.5,0.5) {hasDegree};  
				\draw[->]  (c) -- (f);
				\node (r5) [fontscale=1.5] at (5.7,3.4) {hadAdvisor};     
			\end{tikzpicture}   
			\caption{Removing hasDegree}
			\label{fig:hasDegree}
		\end{subfigure}
	\end{minipage}		
	\caption{Subqueries generated by removing one edge at a time.}
	\label{fig:subquery}
\end{figure*}

\begin{figure}
	\centering
	\begin{minipage}{0.23\textwidth}
		\begin{subfigure}[t]{\textwidth}
			\centering
			\begin{tikzpicture}[scale=1, transform shape]
				\node[itria] (0) at +(0,-0.75) [fontscale=4.00]{SQ1};
				\tikzstyle{every node} = [circle, fill=gray!30]
				\node (1) at +(0,0) {?x};
				\node[fill=none, draw=black] (2) at +(1.5,0) {?y};
				\draw[line width=1.8pt, -, dashed]  (1) -- (2) node [midway,below,fill=none] {R};
				\draw[-]  (1) -- (0) ;     
			\end{tikzpicture}    
			\caption{Type  I} \label{fig:type1}
		\end{subfigure}
	\end{minipage}
	\begin{minipage}{0.23\textwidth}
		\begin{subfigure}[t]{\textwidth}
			\centering
			\begin{tikzpicture}[scale=1, transform shape]
				\node[itria] (0) at +(-2,-0.75) [fontscale=4.00]{SQ1};
				\tikzstyle{every node} = [circle, fill=gray!30]
				\node (1) at +(-2,0) {?x};
				\node[fill=none, draw=black] (2) at +(0,0) {?z};
				\node[fill=none, draw=black] (3) at +(-1,0) {?y};
				\draw[line width=1.8pt, -, dashed]  (1) -- (3) node [midway,below,fill=none] {R};
				\draw[-]  (2) -- (3) node [midway,below,fill=none] {E};  
				\draw[-]  (1) -- (0) ;     
			\end{tikzpicture}  
			\caption{Type II} \label{fig:type2}
		\end{subfigure}
	\end{minipage}
	\begin{minipage}{0.23\textwidth}
		\begin{subfigure}[t]{\textwidth}
			\centering
			\begin{tikzpicture}[scale=1, transform shape]
				\node[itria] (0) at +(0,-0.75) [fontscale=4.00] {SQ1};
				\node[itria] (3) at +(1.5,-0.75) [fontscale=4.00] {SQ2};
				\tikzstyle{every node} = [circle, fill=gray!30]
				\node (1) at +(0,0) {?x};
				\node (2) at +(1.5,0) {?y};
				\draw[line width=1.8pt, -, dashed]  (1) -- (2) node [midway,below,fill=none] {R};
				\draw[-]  (1) -- (0) ;     
				\draw[-]  (2) -- (3) ;     
			\end{tikzpicture}    
			\caption{Type III} \label{fig:type3}
		\end{subfigure}
	\end{minipage}
	\begin{minipage}{0.23\textwidth}
		\begin{subfigure}[t]{\textwidth}
			\centering
			\begin{tikzpicture}[scale=1, transform shape]
				\node[itria] (0) at +(0.75,-0.75) [fontscale=4.00]{SQ1};
				\tikzstyle{every node} = [circle, fill=gray!30]
				\node (1) at +(0,0) {?x};
				\node (2) at +(1.5,0) {?y};
				\draw[line width=1.8pt, -, dashed]  (1) -- (2) node [midway,below,fill=none] {R};
				\draw[-]  (1) -- (0) ;     
				\draw[-]  (2) -- (0) ;     
			\end{tikzpicture}  
			\caption{Type IV} \label{fig:type4}
		\end{subfigure}
	\end{minipage}
	\caption{Types of subqueries}
	\label{fig:category}
\end{figure}

\subsection{Connection Point Annotations}
\label{subsec:annotate}

After identifying all the potential matches of a query, \name annotates the
entities corresponding to their connection points. Each connection point
results in its entity to be annotated with a tuple \textsf{(expRel, dir, sqId,
result, provPoly)}, where \textsf{expRel} is the relation predicate the
connection point expects, \textsf{dir} is the direction of the edge (either
outgoing or incoming), \textsf{sqId} is the identifier of the subquery that the
connection point belongs to, \textsf{result} is the value of the query variable
that needs to be projected, and \textsf{provPoly} is the provenance polynomial
corresponding to the potential match.

For instance, consider the subqueries generated by removing the \labl{coAuthor}
edge in Fig. \ref{fig:coAuthor}. For ease of exposition only two connections
are considered, one for each subquery. The node \lit{Ooi} in the KG in
Fig.~\ref{fig:runningEg} matches one of the connection points, and is annotated
with the tuple (coAuthor, out, $q_{i}$, Ooi, $e_{15}.e_{16}$) as it is
expecting an outgoing edge with label \labl{coAuthor}. Similarly, node
\lit{Gehrke} has the annotation (hadAdvisor, in, $q_{j}$, Ramakrishnan,
$e_{1}.e_{3}$).  Note that for \lit{Gehrke} the \textsf{result} attribute has
value \lit{Ramakrishnan} as it would match the projected variable \lit{?prof}
of the parent query.

\subsection{Subquery Execution Plan}
\label{subsec:planBuilder}

As new edges are added to the KG, it is crucial that the results of all the
subqueries of registered queries are continually and efficiently maintained.
Instead of generating an execution plan each time subquery results have to be
updated, \name generates an execution plan for each subquery at the time of
query registration. These plans are then merged together to form a \emph{global
plan} that tries to exploit shared execution segments. 

\subsubsection*{Local Plan Generation}

In order to generate the execution plan for each subquery, \name adapts AND-OR
trees~\cite{andortree}. These trees are comprised of two types of nodes, AND
and OR nodes. They are present at alternate levels in the tree, i.e., each AND
node has an OR node as its parent and also as children. An AND node represents
a way of deriving the intermediate expression corresponding to its parent OR
node, and vice versa.  These trees effectively capture all possible derivations
of the subquery expressions, with OR nodes corresponding to an intermediate
expression in the evaluation of the final subquery expression.

After the AND-OR tree is constructed, it is required to choose the \emph{best}
query plan of the subquery. A greedy heuristic is employed to choose the best
plan. Each intermediate expression is assigned a score based on its estimated
cardinality. The execution tree is then traversed in a bottom-up manner. At
each level $i$, the node corresponding to the intermediate expression with the
lowest cardinality is greedily chosen. No other node at level $i$ is
considered. At the next level $i + 1$, its parent nodes are explored.  The
traversal continues till the root of the tree is reached.

Consider the parent query shown in Fig.~\ref{fig:query}. Removal of the triple
pattern $\langle ?stud, hadAdvisor,?prof\rangle$  generates the subquery with
triple patterns containing $hasDegree$, $worksIn$ and $coAuthor$ as predicates.
A snippet of partial AND-OR tree for this subquery is shown in Fig.
\ref{fig:plan}. In interests of space, relations $hasDegree$, $worksIn$ and
$coAuthor$ are denoted by $P_{1}$, $P_{2}$ and $P_{3}$ respectively, As shown,
AND nodes on the second level correspond to intermediate expressions $(P_{1}
\bowtie_\text{subject} P_{2})$ and $(P_{2} \bowtie_\text{subject} P_{3})$.
Similarly, the root node corresponds to the final expression of the subquery.
The children of the root node represent the two possible derivations, viz.,
$((P_{2} \bowtie_\text{subject} P_{3}) \bowtie_\text{subject} P_{1})$ and
$((P_{1} \bowtie_\text{subject} P_{2}) \bowtie_\text{subject} P_{3})$.

\subsubsection*{Global Plan Computation}

Once the best plan is chosen for a subquery, it is merged into the global query
plan. The global plan is built to reuse the computation of shared intermediate
expressions across all subqueries.  In a local plan, the subtree rooted at a
given node gives its derivation.  We start traversing the best local plan in a
top-down manner. We do a breadth-first traversal down the best local plan. At
each node $k$, we look for an equivalent expression in the global plan. If an
equivalent expression already exists, it implies that the derivation of the
expression corresponding to node $k$ exists.  So, we do not explore another
derivation of the same expression. If no equivalent expression is found, we add
the node $k$ along with the corresponding equivalent expression to the global
tree and move to the next node. The global plan can have multiple roots. Each
root represents the final expression of one of the subqueries.

\subsubsection*{Cardinality Estimation} 

In order to obtain high-quality execution plans for queries, an accurate
cardinality estimation is crucial. The cardinality estimation used in \name is
specifically designed for RDF knowledge graphs, and captures the neighborhood
of entities in the graph. Specifically, \name models queries as
\emph{star-chain} shapes, i.e., a sequence of star shaped
queries~\cite{meimaris2017extended}.  Based on this, the statistics are
maintained as counts of $(s,o)$ pairs that have the same \emph{characteristic
pair}~\cite{csPair} $P_{c}$, given by:
$P_{c}(S_{c}(s), S_{c}(o)) =$ $\{ (S_{c}(s), S_{c}(o), p) \mid S_{c}(o) \neq
\phi \wedge \exists p: (s,p,o) \in E \}$
where $S_c$ is the characteristic set.  The cardinality estimate of a query,
represented as a sequence of $k$ star-shaped fragments $C_1,\dots,C_k$, is the
sum of the estimates of the pairs of fragments:
\begin{align*}
	card = \sum_{i=1}^{k-1} count(P_{c}(C_{i},C_{i+1}))
\end{align*}

\begin{figure}
	\tikzstyle{andNode} = [ellipse, draw,  
	 text width=3.5em,text centered, minimum height=1.8em,fontscale=0.07,line width=0.6pt]
	\tikzstyle{orNode} = [rectangle, draw,  
	 minimum height=1.5em, fontscale=0.07,line width=0.6pt]
	\centering
	\begin{tikzpicture}[scale=0.5]
		\node[orNode, text centered, font =\footnotesize] (root) at (4,4) {$\{P_{1},P_{3}, P_{2}\}$};
		\node[andNode, font =\footnotesize, text centered] (1) at (2,2) {$(P_{2},P_{3}), P_{1}$};    
		\node[andNode, font =\footnotesize, text centered] (2) at (6,2) {$(P_{1},P_{2}), P_{3}$};    
		\node[orNode, font =\footnotesize, text centered] (3) at (2,0) {$\{P_{2},P_{3}\}$};    
		\node[orNode, font =\footnotesize, text centered] (4) at (6,0) {$\{P_{1},P_{2}\}$};    
		\node[andNode, font =\footnotesize,text width=3em] (5) at (2,-2) {$(P_{1}, P_{2})$};    
		\node[andNode, font =\footnotesize, text width=3em] (6) at (6,-2) {$(P_{2},P_{3})$};    
		\node[orNode, font = \footnotesize] (p2) at (1,-4) {$P_{1}$: hasDegree};    
		\node[orNode, font = \footnotesize] (p1) at (4.8,-4) {$P_{2}$: worksIn};    
		\node[orNode, font=\footnotesize] (p) at (8.5,-4) {$P_{3}$: coAuthor};    
		\draw[->] (p) .. controls (8.75, -0.6) .. (2);
		\draw[->] (p2) .. controls (-0.5, -0.57) .. (1);
		\draw[->] (p) -- (6);
		\draw[->] (p1) -- (6);
		\draw[->] (p1) -- (5);
		\draw[->] (p2) -- (5);
		\draw[->] (5) -- (4);
		\draw[->] (6) -- (3);
		\draw[->] (3) -- (1);
		\draw[->] (4) -- (2);
		\draw[->] (1) -- (root);
		\draw[->] (2) -- (root);
	\end{tikzpicture}
	\captionof{figure}{AND-OR tree. (Boxes represent OR nodes while
	ellipse represent AND nodes.) }
	\label{fig:plan}
\end{figure}

\subsection{Processing \mmp Queries}
\label{subsec:queryChecker}

The above discussed steps involving subquery generation, connection point
annotation and global execution plan generation work for maintaining regular
queries. \mmp queries, however, require special handling.  The presence of more
than one triple with the same predicate is a necessary but not a sufficient
condition to identify a conjunctive query as a \mmp query.  We next use an
example of a non-\mmp query $Q$ in Fig.~\ref{fig:invalid} to highlight the
characteristics of \mmp queries.

\subsubsection*{Characteristics of \mmp Queries}

The query $Q$ has two pairs of triple patterns with the same predicate:
($\langle?x1, P_{4},?x2 \rangle$, $\langle?x5, P_{4}, ?x6\rangle$) and
($\langle?x2, P_{5},C1 \rangle$, $\langle?x6, P_{5}, C2\rangle$).  There are
two important observations:

\begin{itemize}

	\item If $C_{1} \neq C_{2}$ and relation $P_{5}$ is a one-to-one relation
		then triples $\langle?x2, P_{5},C1 \rangle$ and $\langle?x6, P_{5},
		C2\rangle$ cannot be mapped to the same edge of a KG. In other words,
		$?x2$ and $?x6$ cannot be bound to a same node. Propagating this
		constraint backwards, triples $\langle?x1, P_{4},?x2 \rangle$ and
		$\langle?x5, P_{4}, ?x6\rangle$ cannot be mapped to the same edge as
		$?x2$ and $?x6$ cannot be mapped to same node of KG.  So, triple
		patterns participating in two different paths, comprising  of
		one-to-one relations, leading to different bound variables cannot be
		mapped to same edges of KG. Such a query cannot be \mmp. For, instance,
		paths from $?x1$ to $C1$ and $?x5$ to $C2$ has same relation sequence
		$P_{4}$, $P_{5}$ but none of the involved triple patterns can be mapped
		to same KG edge. Thus, $Q$ is not a \mmp query.

	\item Suppose relations $P_{1},P_{2}$ and $P_{3}$ are same, say
		\labl{isLocatedIn}.  Then node $?x1$ and $?x5$ cannot be mapped to same
		node of KG as a location cannot be located in itself. Thus, triples
		$\langle?x1, P_{4},?x2 \rangle$ and $\langle?x5, P_{4}, ?x6\rangle$
		cannot be mapped to same edge.  Therefore, a query cannot have a \pmm
		involving variables $?x$ and $?z$, if a path $P_{i}, \dots, P_{j}$ from
		$?x$ to $?z$ is comprised of the same \emph{asymmetric} predicate,
		i.e., $P_{i} = \dots = P_{j} = P$ where $(a P b \implies \neg b P a)$.

\end{itemize}

\subsubsection*{Identification of \mmp Queries}

Based on the above two observations, we identify \mmp queries and the
corresponding predicate $P$ shared by the triple patterns, $t_{i}, \dots,
t_{j}$, that can match to the same edge. A \mmp query can have both \pmo and
\pmm. A \pmo would satisfy one of the subqueries (Lemma~\ref{obs:ob1}) but a
\pmm would not satisfy any subquery before an appropriate edge $e$ is inserted
(Lemma~\ref{obs:ob2}). However, when an appropriate edge is inserted, the \pmm
would satisfy all the subqueries and the parent \mmp query simultaneously
(Lemma~\ref{obs:ob3}). Thus, to identify new matches to the parent \mmp query,
we examine the new results of a subquery.  Instead of examining all subqueries,
only those that are generated by removing a triple pattern that shares the same
predicate, are enough to identify a \pmm.

\subsection{Processing to Support Deletions}
\label{subsec:indexBuilder}

An edge deletion requires not only the updates to the parent query results, but
also to the existing potential matches and the connection point annotations in
the KG. In the running example (Fig.~\ref{fig:runningEg}), if edge $e_{15}$ is
deleted, then the node \lit{Ooi} is no longer a connection point. \name
constructs two inverted index structures: \textsf{edgeToResult} and
\textsf{edgeToCP} to efficiently support this. The \textsf{edgeToResult} index
maintains a mapping from an edge to the answer set items containing the edge in
their how provenance (i.e., the presence of the edge contributed to the item
being in the answer set). The \textsf{edgeToCP} index simply maintains the
correct state of all the connection points corresponding to the edge. 

To build the \textsf{edgeToResult}, the following two sets of information for
each query $Q_{i}$ are maintained: (a)~$R_i$: the set of result items of query
$Q_i$ along with their corresponding provenance polynomials, and (b)~$E_i$: the
set of all edges involved in the provenance polynomials of all the result items
$R_i$.

Based on these, an index is built on edges $E_i$ that appear in the result set
of at least one query.  The index entry corresponding to an edge $e$ consists
of the set $\{ \langle Q_{i}, L_{i} \rangle \}$ where $L_{i}$ is the list of
result items such that $e$ appears in its associated provenance polynomial.

Similarly, the \textsf{edgeToCP} index is built by considering the provenance
polynomials of the potential matches. Each participating edge is mapped to a
set of connection points which are annotated with those provenance polynomials.

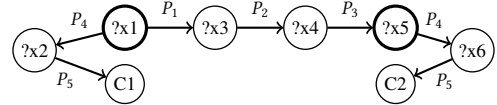
\begin{figure}
	\centering
	\begin{tikzpicture}[scale=0.6, transform shape]
		\tikzstyle{every node} = [circle, draw=black,fontscale=1.5,line width=0.5pt]
		\node (x3) at +(2,2) {?x3};
		\node (x2) at +(-2,1.5) {?x2};
		\node[line width=1.2pt] (x1) at +(0,2) {?x1};
		\node (x4) at +(4,2) {?x4};
		\node[line width=1.2pt] (x5) at +(6,2) {?x5};
		\node (x6) at +(7.7,1.5) {?x6};
		\node (x21) at +(0,0.75) {C1};
		\node (x61) at +(6,0.75) {C2};
		\tikzstyle{every node} = [circle, draw=none,fontscale=1.5]
		\draw[->,thick]  (x1) -- (x3) node [midway,above,fill=none] {$P_{1}$};   
		\draw[->,thick]  (x3) -- (x4) node [midway,above,fill=none] {$P_{2}$};   
		\draw[->,thick]  (x5) -- (x6) node [midway,above,fill=none] {$P_{4}$};   
		\draw[->,thick]  (x1) -- (x2) node [midway,above,fill=none] {$P_{4}$};   
		\draw[->,thick]  (x4) -- (x5) node [midway,above,fill=none] {$P_{3}$};   
		\draw[->,thick]  (x2) -- (x21) node [midway,below,pos=0.2,fill=none] {$P_{5}$};   
		\draw[->,thick]  (x6) -- (x61) node [midway,below,pos=0.2,fill=none] {$P_{5}$};   
	\end{tikzpicture}    
	\captionof{figure}{Example of a non-\mmp query.}
	\label{fig:invalid}
\end{figure}

\section{Update Processing}
\label{sec:update}

\subsection{Insertion}

When a new fact is added to the KG, we are first required to find the affected
queries and then update the indexes and annotations for handling subsequent
updates.  Suppose that edge $e(u,P,v)$ is added to a KG. We start by fetching
relevant annotations of $u$ that have $P$ as \textsf{expRel} and ``out'' as
\textsf{dir}. Suppose, the relevant annotation has query id \textsf{sqId} as
$q_{i}$.  Recall that each annotation corresponds to a \pmo of one of the
subqueries. Based on the type of subquery $q_{i}$, we proceed as follows:

\begin{itemize}

	\item \emph{Type I:} If $q_{i}$ is a Type I subquery, then we do not need
		to perform any check further as the insertion of $e$ with expected
		relation $P$ has completed the potential match. For example, in
		Fig.~\ref{fig:kg}, if edge $\langle Sarawagi, worksIn, IITB \rangle$ is
		inserted, then based on the fact that $Sarawagi$ is expecting an
		outgoing edge label $worksIn$ is enough to confirm the transition of
		the corresponding potential matches.

	\item \emph{Type II:} Recall that this class has one subquery $SQ_{1}$  of
		size $n-2$ and an single edge with predicate $E$. We need a simple
		lookup to check if node $v$ has an correct incident edge with label
		$E$. The presence of such edge leads to a parent query answer.
		
	\item \emph{Type III:} We examine the annotation of $v$ as well.  If $v$
		has an annotation with $P$ as \textsf{expRel} and ``in" as
		\textsf{dir}, and query id of complimentary subquery of subquery
		$q_{i}$ as \textsf{sqId}. If the required \textsf{sqId} is present,
		then $e$ changes the potential match to an actual match.  For instance,
		the connection points $Gehrke$ and $Ooi$ in Fig. \ref{fig:kg}
		correspond to a pair of complimentary subqueries.  Therefore, if the
		edge $\langle Ooi, coAuthor, Gehrke \rangle$ is added to the KG, it
		would result into a new answer $\left( Ramakrishnan,Ooi\right)$.

	\item \emph{Type IV:} We check the annotation of $v$ and if we find the
		same subquery id as that in $u$ and with \textsf{expRel} as $P$ and
		``in'' as \textsf{dir}, we can conclude that $e$ has completed the
		potential match.

\end{itemize}

The above checks are for \pmo only. If the predicate $P$ of the newly added
edge is involved in any of the \mmp queries, then it needs special handling.
Before reporting new matches to a \mmp query, the subquery results are updated.
The global query plan is browsed in a bottom-up manner and the incremental
results of all the intermediate expressions involving the predicate $P$ are
computed.  Consequently, we find the incremental results of all the concerned
subqueries. Recall that each \mmp query is associated with one of its
subqueries $S$. The subquery was chosen in a manner that if the newly added
edge $e(u,P,v)$  matches the missing triple pattern of the subquery, then all
the new answers of the subquery would match the parent query, and, thus, would
be reported as parent query answers. 

The new results also correspond to new connection points which are then
annotated.  Both the indexes \textsf{edgeToResult} and \textsf{edgeToCP} are
also updated.

\subsection{Deletion}

The deletion algorithm is based on the standard \emph{filter-and-refine}
paradigm where candidate queries are first generated and later confirmed.
Assume that an edge $e_{d}$ is deleted. First, we identify the provenance
polynomials in which this edge is involved and the corresponding results.  This
is done by fetching the entry corresponding to $e_{d}$ in the
\textsf{edgeToResult} index. The next task is to confirm if these results
actually get \emph{affected} by the edge deletion.  For this, the polynomial is
evaluated with $e_{d}$ set to $0$ and the other edge symbols remaining at $1$.
An answer vanishes if and only if all the monomials corresponding to it get
evaluated to $0$.  In other words, the answer no longer remains valid. Since
the connection points are annotated with the provenance polynomials, they get
updated as well.  Finally, the entry of $e_{d}$ is dropped from both the
indexes \textsf{edgeToResult} and \textsf{edgeToCP}.

Suppose edge $e_{14}: \langle$\lit{Ooi, hadAdvisor, Stonebraker}$\rangle$ is
deleted from the example KG in Fig.~\ref{fig:runningEg}. We evaluate the
provenance polynomial  of the query answer, given in
Table~\ref{tab:queryresults}, by assigning $0$ to $e_{14}$ and $1$ to the rest.
The resultant value would be $1$ which indicates a valid derivation
$\{e_{2},e_{3},e_{6},e_{8},e_{17}\}$ of the answer exists and, hence,
$\langle$\lit{Stonebraker, Ramakrishnan}$\rangle$ is still an answer. However,
if the edge $e_{2}$ is deleted then the polynomial would collapse to $0$,
denoting no valid derivation exists, i.e., the pair is no longer a valid
answer.

\subsection{Proof of Correctness}

\subsubsection*{Insertion}

We assume a single edge insertion and argue the correctness for the two types
of queries separately.  Since maintenance is done after every insertion
operation, the correctness extends to multiple insertions.

Lemma~\ref{obs:ob1} shows that a \pmo of a regular query will be of size
exactly $n-1$ and will satisfy a subquery.  Hence, once a \pmo is identified
correctly, and an edge insertion completes it, it must match the complete
parent query.  The corresponding subgraph match and its provenance polynomial
is updated to reflect the current state.

A \mmp query may have a \pmm.  If it has a \pmm, using Lemma~\ref{obs:ob3}, we
see that if an edge insertion satisfies the subqueries, it must also satisfy
the parent query.  Thus, once a \pmm is identified correctly in a \mmp query,
the correctness is ensured since the subquery results are also updated.

\subsubsection*{Deletion}

If after the deletion of an edge $e$, a provenance polynomial evaluates to
$0$, it implies that $e$ was present in each corresponding derivation
(monomial). Conversely, as long as a derivation exists that does not involve
$e$, its corresponding monomial will continue to evaluate to $1$, and the
whole polynomial will be non-zero.  If $e$ is not involved in any
derivation of any answer, then if will not have any entry in
\textsf{edgeToResult} index and there will be no impact on the query results.

\section{Experimental Evaluation}
\label{sec:expts}

\subsection{Setup}
\label{subsec:setup}

\subsubsection{Datasets}

We consider two widely used benchmark datasets.

\begin{itemize}[nosep,leftmargin=*]

	\item {YAGO2 \cite{yago2,hoffart2011yago2}:} It is an automatically built
		ontology gathering facts from sources like
		Wikipedia, GeoNames, etc.
		\comment{It has $\sim23$ million facts with $78$
		unique edge labels over $\sim8.8$ million real-world entities.}
	
	\item {DBpedia3.9 \cite{dbpedia}:} It is the structured
		counterpart of Wikipedia.
		\comment{It comprises of $\sim117$ million facts
		over $\sim32$ million real-world entities. The number of
		different relationships among the entities is $\sim53000$.}
	
\end{itemize}	

Table~\ref{tab:dataset} summarizes the statistics of these two datasets.

\begin{table}[t]
	\centering
	\caption{Statistics of datasets}
	\label{tab:dataset}
	\resizebox{\columnwidth}{!}{%
	\begin{tabular}{ccccccc}
		\toprule
		\bf Dataset & Vertices & Edges & Predicates & Queries & Avg. Query Size & Subqueries \\
		\midrule
		\yago	 & 8.8M	& 23M & 78 & 4 & 6.25 & 26 \\
		\dbpedia	&	32M	& 117M & 53K & 215 & 3.90 & 879 \\
		\bottomrule
	\end{tabular}
	}
\end{table}

\subsubsection{Query Sets}

For \yago, we used a set of queries on which the RDF-3X was originally
validated~\cite{neumann2010rdf}. We chose $4$ queries (ids A1, A2, B1, B2) that
are fairly large and complex.

For \dbpedia, we built the query workload from the real world queries available
from 2014 USEWOD (\url{https://eprints.soton.ac.uk/379401/}) query logs.  This
query log is a collection of over $253$K actual \texttt{http} requests to the
\dbpedia SPARQL endpoint.  Out of these, $12$K are SPARQL SELECT queries. After
filtering out nested, ordering and aggregate queries, we obtained a set of
$5,490$ queries.  We further removed queries of size $1$, those which were not
parsed by a SPARQL parser, and duplicates to obtain the final workload of $215$
queries.

\subsubsection{KG Update Workload}

For each KG, we generated insertion workloads in the following manner.
Starting from the original graph, we randomly selected a pair of vertices not
yet connected, and connected them with a randomly selected predicate. Note that
we did not consider any predicate for insertion which is \emph{not} part of at
least one query, since they do not affect the results of any query and can be
trivially answered by a simple lookup. For deletion workloads, we randomly
deleted edges of the predicates that are involved in \emph{at least one} of the
queries.  Using this procedure, we generated workloads of size $10000$ updates
each.

\subsubsection{Implementation}

We conducted all the experiments on a machine with 32-core 2.1 GHz CPU, 512 GB
RAM, and 1 TB hard drive. Our implementation was single-threaded and in Java.
We chose \neo 3.3.9 to store the knowledge graphs due to its capability to
model property graphs. \name codebase is publicly available.
\footnote{\url{https://github.com/gaurgarima/HUKA}}.

\subsubsection{Baseline Systems}

Wylot et al.~\cite{tripleProv} highlighted the shortcomings of \emph{named
graphs} to support the provenance polynomials and devised \tripleprov to
support provenance computation over a customized RDF-engine
\textsf{dipLODocus}~\cite{diplo}.  We contacted the authors for their code but
unfortunately, they responded that \tripleprov is no longer available.  We,
therefore, did a best effort implementation of \tripleprov on top of property
graphs supported by a leading graph database engine \neo and used it as a
baseline.  We are calling this baseline as \neo.  We also used two existing
frameworks that enable provenance-aware query processing on
\textsf{PostgreSQL}, a versatile middle-ware \gprom~\cite{arab2018gprom} and a
recently introduced \textsf{PostgreSQL} extension
\provsql~\cite{senellart2018provsql}.  For the \textsf{PostgreSQL} based
baseline systems, we modeled the KG as relational database by grouping all the
triples of a particular predicate in a table, i.e., for each predicate $P$ of
the KG we construct a table $R_{P}$ and populate it with the triples having
predicate $P$. We also translated SPARQL queries into SQL queries by
interpreting each pair of triple sharing either subject or object with each
other as a join operation on the corresponding predicate relations. 

\subsection{Analysis of Query Processing}

\begin{table}[t]
	\caption{Statistics of query size partitions in \dbpedia.}
	\label{tab:dbpediaSizeStat}
	\centering
	\resizebox{\columnwidth}{!}{%
	\begin{tabular}{cccccc}
		\toprule
		\bf Query & \bf No. of & \bf No. of & \bf \#Intermediate & \bf Global & \bf Coverage \\ 
		\bf set & \bf parent & \bf subqueries & \bf expressions & \bf tree & \\ 
		& \bf queries & & \bf (w/o global plan) & \bf size & \\ 
		\midrule
		$QS_{3}$ & 77 & 231 &  462 & 158 &  2.164  \\ 
		$QS_{4}$ & 88 & 352 & 1056 & 377 & 6.982  \\ 
		$QS_{5}$ & 39 & 224 &  896 & 263 & 9.393  \\ 
		$QS_{6}$ & 10  & 65 & 325  & 112 &  5.895  \\ 
		\midrule
		$entire$ & 215 & 879 & 2781  &  866 & 7.664  \\ 
		\bottomrule
	\end{tabular}
	}
\end{table}

\subsubsection{Subquery Generator}

The subquery generator produced $26$ and $879$ subqueries out of $4$ and $215$
parent queries present in \yago and \dbpedia query sets respectively. Based on
the results of these subqueries, connection points were identified. For \yago,
subqueries resulted into $164$K connection points whereas \dbpedia queries
generated $6.5$M connection points.

\subsubsection{Global Execution Plan}

The global execution tree for \yago subqueries has $80$ nodes while for
\dbpedia, it has $866$ nodes.  Each node in the global execution tree
corresponds to an \emph{intermediate expression}.  To analyze the queries
better, we partitioned the query set of \dbpedia based on the query size. In
Table \ref{tab:dbpediaSizeStat}, the query set $QS_{i}$ is the set of \dbpedia
queries of size $i$.  The \emph{entire} query set corresponds to the full query
set.

To confirm our claim that subqueries are likely to overlap and, thus, the
global plan promotes re-usability, we compare the number of intermediate
expressions in the global plan to the total number of intermediate expressions
required to compute the final expression of all the subqueries. To compute the
latter, we sum up the sizes of all the subqueries as the size, i.e., number of
triple patterns, gives the number of joins required to compute it which, in
turn, correspond to the number of intermediate expressions. We have reported
the total number of intermediate expressions in the absence of a global plan in
Table \ref{tab:dbpediaSizeStat}. The global plan has cut down the intermediate
expression computation considerably.  The total number of intermediate
expressions is reduced by close to $70\%$.  (The column ``coverage'' will be
discussed in Sec.~\ref{sec:closer}.)

\subsubsection{Space Usage}

The \textsf{edgeToResult} inverted index sizes for \yago and \dbpedia are
$45$MB and $322$MB respectively while the \textsf{edgeToCP} index sizes are
$41$MB and $995$MB respectively.

\subsubsection{Time Consumption}

We took an average time of $6.9$ minutes and $10$ minutes to register \dbpedia
and \yago queries respectively.  Majority of the time while registering a query
is spent on execution plan generation as we exhaustively generate all possible
derivations of a subquery.  Larger the query size, more is the number of
derivations.  Since \yago queries are fairly large with an average size of
$6.25$ triples per query, whereas \dbpedia queries size has an average of
$3.90$ triples per query, it requires more time to register \yago queries than
\dbpedia.

It needs to be noted that queries are registered only once, though.  The
average time spent per update operation subsequently is much less as reported
next.

\subsection{Analysis of Update Processing}

\begin{table}[t]
	\caption{Comparison of average updation times.}
	\label{tab:baseline}
	\centering
	{\small
	\begin{tabular}{c|r|rrr}
		\toprule
		\bf Dataset & \name & \gprom & \provsql & \neo \\
		\midrule
		\yago  & 0.119\,s & 25.121\,s & 75.657\,s & 5.709\,s \\ 
		\dbpedia & 1.252\,s & 5.217\,s & 6.870\,s & 99.318\,s \\
		\bottomrule
	\end{tabular}
	}
\end{table}

\subsubsection{Baseline Comparison}

The running times (averaged over $5$ runs) per update operation taken by our
system \name and the various baselines on the two datasets are reported in
Table~\ref{tab:baseline}.  For \name, we recorded the average total time taken
to \emph{update} the parent queries \emph{and maintain} the supporting
structure. As none of the baseline systems support incremental update of query
provenance, we computed the average total time taken to execute all the
\emph{relevant} queries for a particular update request. The relevant queries
are those in which the updated edge predicate is present.

On \yago, \name takes less than a second to update the query results and their
provenance polynomials, whereas for \dbpedia, the time taken is slightly more.
Importantly, our approach is around $4$ to $48$ times faster than the nearest
baselines.

Since the query size directly translates to number of joins, \gprom and
\provsql require lesser time for \dbpedia.  However, for \name, the number of
affected queries for \dbpedia are generally higher as the total number of
queries ($4$ in \yago versus $215$ in \dbpedia) are higher.  This overshadows
the impact of lesser number of joins per query and, consequently, \name
requires more time for \dbpedia than \yago.  Unlike relational data, \neo uses
graph data organization and stores logically related data nearby. This results
in efficient graph traversal and, therefore, \neo requires considerably lesser
time for \yago since it is mostly impacted by the number of queries.

\begin{table}[t]
	\caption{\name under varying workloads.}
	\label{tab:workload}
	{\small
	\begin{tabular}{c|ccccc}
		\toprule
		\multirow{2}{*}{\bf Dataset} & Deletion & Deletion  & \multirow{2}{*}{Balanced} & Insertion & Insertion \\
		                             & -Heavy   & -Moderate & & -Moderate & -Heavy \\
		\midrule
		\yago    & 0.062\,s & 0.091\,s & 0.126\,s & 0.146\,s & 0.169\,s \\
		\dbpedia & 1.126\,s & 1.144\,s & 1.316\,s & 1.344\,s & 1.329\,s \\
		\bottomrule
	\end{tabular}
	}
\end{table}

\subsubsection{Varying Workloads}

We next study in detail the effect of different types of updates: deletion
versus insertion.  We generated multiple workloads (as described in
Sec.~\ref{subsec:setup}) of size $10000$ updates each with controlled ratio of
operations.  We used $5$ configurations of deletion to insertion ratios:
\textit{Insertion-Heavy} (1:9), \textit{Insertion-Moderate} (3:7),
\textit{Balanced} (5:5), \textit{Deletion-Moderate} (7:3), and
\textit{Deletion-Heavy} (9:1).  The average times are reported in
Table~\ref{tab:workload}.  For both the datasets, the average total time
gradually increases as we move from \emph{deletion-heavy} to
\emph{insertion-heavy} workload. The trend is expected as insertion is
computationally heavier than deletion. 

\subsubsection{Closer Look}
\label{sec:closer}

Next, to get a better insight of our system, we decided to have a closer look
at its different steps.  We chose the larger query set \dbpedia and the query
set partitions outlined in Table~\ref{tab:dbpediaSizeStat}. We treat each query
set $QS_i$ individually, i.e., built separate query execution plans, annotated
separate KGs, and constructed separate indexes for each query set.  

\begin{figure}[t]
\centering
	\includegraphics[width=0.75\columnwidth]{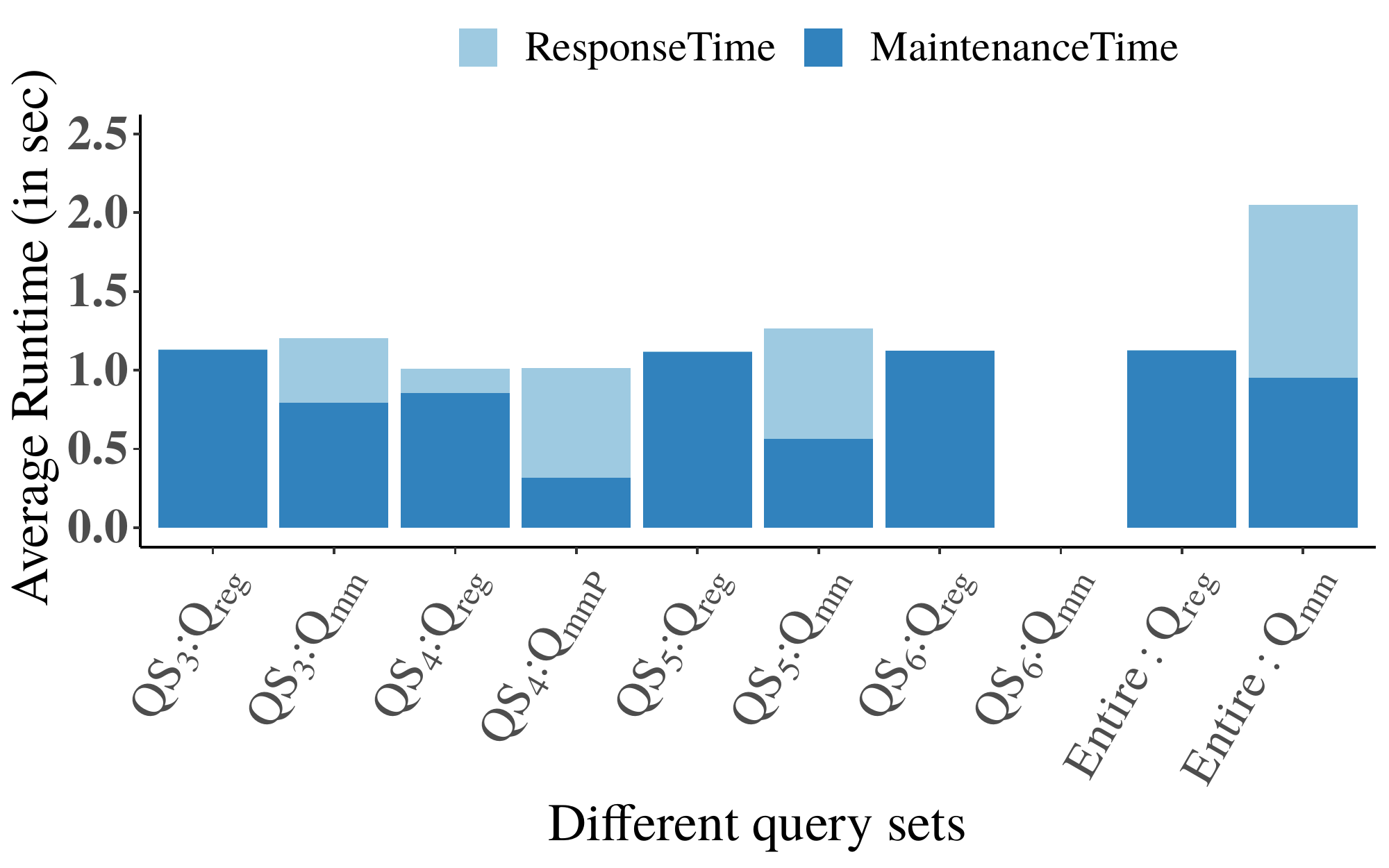}
	\caption{Time split of different subtasks for \dbpedia.}
	\label{fig:split}
\end{figure}

We divide the total running time into $2$ components: (a)~\emph{response time},
defined as the time taken to update the parent query result(s), and
(b)~\emph{maintenance time}, defined as the time taken to maintain the data
structures so as to ensure correctness of subsequent updates.  We can update
result of regular parent queries just by examining the relevant connection
points, whereas to update a \mmp parent query we need to update the subquery
result as well.  For regular queries, maintenance time is the sum of time
consumed to update subqueries, annotate graph and update inverted indexes.
However, for \mmp queries, maintenance time involves annotating graph and
update indexes.  The time splits in terms of percentage of the total time for
the \emph{balanced} workload are shown in Fig.~\ref{fig:split}.  (There is no
\mmp query of size $6$.) Across query sets regular queries exhibit lower
response time as compared to \mmp queries.  This happens as \mmp queries more
background work before the query results can be updated.  Overall, regular
queries require lesser absolute time to complete.

Importantly, the time to update subquery results was agnostic to the number of
subqueries of a query set $QS_{i}$.  To understand this, we closely inspected
the global plans of each query set $QS_{i}$. Intuitively, higher the number of
intermediate expressions to compute, higher is the subquery result updation
time. To quantify this, we defined \emph{coverage} of a global plan as the
average number of intermediate expressions required to be computed when an edge
is inserted. In other words, it is the ratio of the total number of non-leaf
nodes in the global tree to the total number of unique predicates involved in
the corresponding query set.  The \emph{coverage} values for $QS_{i}$ are given
in Table~\ref{tab:dbpediaSizeStat}. Thus, higher the \emph{coverage} of a query
plan, higher is the time to update the subquery results.

\section{Conclusions and Future Work}
\label{sec:conc}

In this paper, we have proposed a system to efficiently maintain provenance of
results of standing queries in knowledge graphs that are dynamic in nature.  We
used decomposition of queries into subqueries and inverted indices to quickly
update the answers and their provenance polynomials.  Experiments show the
effectiveness in large real-life knowledge graphs.

In future, we would like to work with more generic SPARQL queries that include
aggregates, nested queries, ordering, etc.


\bibliographystyle{ACM-Reference-Format}
\balance
\bibliography{../papers}

\end{document}